\documentclass[11pt]{article}
\usepackage{amsthm}
\usepackage{amsmath}
\usepackage{amssymb}
\usepackage{mathptmx}
\usepackage{a4wide}

\newtheorem{thm}{Theorem}[section]
\newtheorem{lem}[thm]{Lemma}
\newtheorem{clm}[thm]{Claim}
\newtheorem{cor}[thm]{Corollary}
\newtheorem{probl}[thm]{Problem}
\newtheorem{obs}[thm]{Observation}
\newtheorem{conj}{Conjecture}
\newtheorem{dfn}[thm]{Definition}
\newtheorem{ex}[thm]{Example}
\newtheorem{rem}[thm]{Remark}

\newcommand{\+}{\oplus} % sum modulo-2
\renewcommand{\vec}[1]{\boldsymbol{#1}} % vectors in bold italics
\newcommand{\Alpha}{\vec{\alpha}}
\newcommand{\Beta}{\vec{\beta}}
\newcommand{\x}{\vec{x}}
\newcommand{\y}{\vec{y}}
\newcommand{\z}{\vec{z}}
\newcommand{\e}{\vec{e}}
\renewcommand{\a}{\vec{a}}
\renewcommand{\b}{\vec{b}}
\renewcommand{\c}{\vec{c}}
\renewcommand{\v}{\vec{v}}

\newcommand{\p}{\vec{p}}

\newcommand{\m}{\vec{m}}

\newcommand{\Deg}[2]{\mathrm{deg}_{#2}(#1)} % degree of circuits
\newcommand{\nulis}{\vec{0}} % all-0 vector
\newcommand{\vienas}{\vec{1}} % all-1 vector
\newcommand{\cov}[1]{\mathrm{cov}(#1)} % covering by rows/colums
\newcommand{\rk}[1]{\mathrm{rk}(#1)} % rank
\newcommand{\minrk}[1]{\mathrm{mr}(#1)} % min-rank
\newcommand{\maxrk}[1]{\mathrm{Mr}(#1)} % min-rank
\newcommand{\wminrkcol}[1]{\mathrm{mr}_{\mathrm{col}}(#1)}
\newcommand{\wminrkrow}[1]{\mathrm{mr}_{\mathrm{row}}(#1)}
\newcommand{\rig}[2]{{\cal R}_{#1}(#2)} % rigidity
\newcommand{\cl}[1]{\omega(#1)} % clique number
\newcommand{\lin}[1]{\mathrm{lin}(#1)} % max linear solution
\newcommand{\opt}[1]{\mathrm{opt}(#1)} % maximal solution
\newcommand{\proj}[2]{#1\!\!\upharpoonright_{#2}} % projection
\newcommand{\Ball}[1]{\mathrm{Ball}(#1)} % Hammin ball
\newcommand{\skl}[2]{\langle#1,#2\rangle} % sclar product
\newcommand{\gf}{\{0,1\}} % Galois field
\renewcommand{\ker}[1]{\mathrm{ker}(#1)} % kernel with parenthesis
\newcommand{\spnm}[1]{\widehat{#1}} % span of a matrix
\newcommand{\width}[1]{\mathrm{width}(#1)} % circuit width
\newcommand{\match}[1]{\mathrm{m}(#1)} % maximal matching
\begin{document}

\title{Min-Rank Conjecture for Log-Depth Circuits
\thanks{Research of both authors supported by a DFG grant SCHN~503/4-1.}}

\author{S. Jukna \ \ \ \ G. Schnitger\\[1ex]
{\small University of Frankfurt, Institute of Computer Science}\\
{\small Frankfurt am Main, Germany}}

\date{} 

\maketitle

  \begin{abstract}
    A completion of an $m$-by-$n$ matrix $A$ with entries in
    $\{0,1,\ast\}$ is obtained by setting all $\ast$-entries to
    constants $0$ and $1$.  A system of semi-linear equations over
    $GF_2$ has the form $M\x=f(\x)$, where $M$ is a completion of $A$
    and $f:\gf^n\to\gf^m$ is an operator, the $i$th coordinate of
    which can only depend on variables corresponding to $\ast$-entries
    in the $i$th row of~$A$. We conjecture that no such system can
    have more than $2^{n-\epsilon\cdot\minrk{A}}$ solutions, where
    $\epsilon>0$ is an absolute constant and $\minrk{A}$ is the
    smallest rank over $GF_2$ of a completion of~$A$.  The conjecture
    is related to an old problem of proving super-linear lower bounds
    on the size of log-depth boolean circuits computing linear
    operators $\x\mapsto M\x$.  The conjecture is also a
    generalization of a classical question about how much larger can
    non-linear codes be than linear ones.  We prove some special cases
    of the conjecture and establish some structural properties of
    solution sets.
  \end{abstract}

\section{Introduction}

One of the challenges in circuit complexity is to prove a super-linear
lower bound for log-depth circuits over $\{\&,\lor,\neg\}$ computing
an explicitly given boolean operator $f:\{0,1\}^n\to\{0,1\}^n$.
Attempts to solve it have led to several weaker problems which are
often of independent interest.  The problem is open even if we impose
an additional restriction that the depth of the circuit is~$O(\log
n)$. It is even open for \emph{linear} log-depth circuits, that is,
for log-depth circuits over the basis $\{\oplus,1\}$, in spite of the
apparent simplicity of such circuits. It is clear that the operators
computed by linear circuits must also be linear, that is, be
matrix-vector products $\x\to M\x$ over the field
$GF_2=(\gf,\+,\cdot)$,

An important result of Valiant \cite{valiant} reduces the lower bounds
problem for log-depth circuits over $\{\&,\lor,\neg\}$ to proving
lower bounds for certain depth-$2$ circuits, where we allow
\emph{arbitrary} boolean functions as gates.

\subsection{Reduction to depth-$2$ circuits}

A depth-$2$ circuit of \emph{width} $w$ has $n$ boolean variables
$x_1,\ldots,x_n$ as input nodes, $w$ arbitrary boolean functions
$h_1,\ldots,h_w$ as gates on the middle layer, and $m$ arbitrary
boolean functions $g_1,\ldots,g_m$ as gates on the output layer.
Direct input-output wires, connecting input variables with output
gates, are allowed. Such a circuit computes an operator
$f=(f_1,\ldots,f_m):\gf^n\to\gf^m$ if, for every $i=1,\ldots,m$,
\[
f_i(\x)=g_i(\x,h_1(\x),\ldots,h_w(\x))\,.
\]
The \emph{degree} of such a circuit is the maximum, over all output
gates $g_i$, of the number of wires going directly from input
variables $x_1,\ldots,x_n$ to the gate $g_i$.  That is, we ignore the
wires incident with the gates on the middle layer.  Let $\Deg{f}{w}$
denote the smallest degree of a depth-$2$ circuit of width $w$
computing~$f$.

It is clear that $\Deg{f}{n}=0$ for $f:\{0,1\}^n\to\{0,1\}^n$: just
put the functions $f_1,\ldots,f_n$ on the middle layer. Hence, this
parameter is only nontrivial for $w<n$.  Especially interesting is the
case when $w=O(n/\ln\ln n)$ (see also  Theorem~2.2 in  \cite{PRS}
for more details):

\begin{lem}[Valiant \cite{valiant}]\label{lem:valiant}
  If $\Deg{f}{w}=n^{\Omega(1)}$ for $w=O(n/\ln\ln n)$, then the
  operator $f$ cannot be computed by a circuit of depth $O(\ln n)$
  using $O(n)$ constant fan-in gates.
\end{lem}

Recently, there was a substantial progress in proving lower bounds on
the \emph{size} of (that is, on the total number of wires in)
depth-$2$ circuits. Superlinear lower bounds of the form
$\Omega(n\log^2n)$ were proved using graph-theoretic arguments by
analyzing some superconcentration properties of the circuit as a graph
\cite{DDPW,Pippenger1,Pippenger2,PS,Pudlak,AP,PRS,RT,RS}.  Higher
lower bounds of the form $\Omega(n^{3/2})$ were proved using
information theoretical arguments~\cite{cherukhin,juk}. But the
highest known lower bound on the \emph{degree} of width $w$ circuits
has the form $\Omega((n/w)\ln (n/w))$~\cite{PRS}, and is too weak to
have a consequence for log-depth circuits.

A natural question therefore was to improve the lower bound on the
degree at least for \emph{linear} circuits, that is, for depth-$2$
circuits whose middle gates as well as output gates are linear boolean
functions (parities of their inputs).  Such circuits compute linear
operators $\x\mapsto M\x$ for some $(0,1)$-matrix $M$; we work
over~$GF_2$.  By Valiant's reduction, this would give a super-linear
lower bound for log-depth circuits over $\{\oplus,1\}$.

This last question attracted attention of many researchers because of
its relation to a purely algebraic characteristic of the underlying
matrix $M$---its rigidity. The \emph{rigidity} $\rig{M}{r}$ of a
$(0,1)$-matrix $M$ is the smallest number of entries of $M$ that must
be changed in order to reduce its rank over $GF_2$ to $r$. It is not
difficult to show (see \cite{valiant}) that any linear depth-$2$
circuit of width $w$ computing $M\x$ must have degree at least
$\rig{M}{w}/n$: If we set all direct input-output wires to $0$, then
the resulting degree-$0$ circuit will compute some linear
transformation $M'\x$ where the rank of $M'$ does not exceed the width
$w$. On the other hand, $M'$ differs from $M$ in at most $dn$ entries,
where $d$ is the degree of the original circuit. Hence,
$\rig{M}{w}\leq dn$ from which $d\geq \rig{M}{w}/n$ follows.

Motivated by its connection to proving lower bounds for log-depth
circuits, matrix rigidity (over different fields) was considered by
many authors, \cite{razborov,alon,PV,Friedman,Pudlak,PRS,SSS,RKh,
  lokam1,lokam2,paturi,Wolf} among others. It is therefore somewhat
surprising that the highest known lower bounds on $\rig{M}{r}$ (over
the field $GF_2$), proved in \cite{Friedman,SSS} also have the form
$\Omega((n^2/r)\ln(n/r))$, resulting to the same lower bound
$\Omega((n/w)\ln (n/w))$ on the degree of linear circuits as that for
general depth-$2$ circuits proved in~\cite{PRS}.  This phenomenon is
particularly surprising, because general circuits may use
\emph{arbitrary} (not just linear) boolean functions as gates.  We
suspect that the absence of higher lower bounds for linear circuits
than those for non-linear ones could be not just a coincidence.

\begin{conj}[Linearization conjecture for depth-$2$ circuits]\label{conj:gen}
  Depth-$2$ circuits can be linearized. That is, every depth-$2$
  circuit computing a linear operator can be transformed into an
  equivalent \emph{linear} depth-$2$ circuit without substantial
  increase of its width or its degree.
\end{conj}

If true, the conjecture would have important consequences for
log-depth circuits. Assuming this conjecture, any proof that every
depth-$2$ circuit of width $w=O(n/\ln\ln n)$ with unbounded fan-in
parity gates for a given linear operator $M\x$ requires
degree~$n^{\Omega(1)}$ would imply that $M\x$ requires a super-linear
number of gates in any log-depth circuit over $\{\&,\lor,\neg\}$. In
particular, this would mean that proving high lower bounds on matrix
rigidity is a much more difficult task than assumed before: such
bounds would yield super-linear lower bounds for log-depth circuits
over a general basis $\{\&,\lor,\neg\}$, not just for circuits over
$\{\oplus,1\}$.

As the first step towards Conjecture~\ref{conj:gen}, in this paper we
relate it to a purely combinatorial conjecture about partially defined
matrices---the \emph{min-rank conjecture}, and prove some results
supporting this last conjecture.  This turns the problem about the
linearization of depth-$2$ circuits into a problem of Combinatorial
Matrix Theory concerned with properties of completions of partially
defined matrices (see, e.g., the survey \cite{Jo90}).  Hence, the
conjecture may also be of independent interest.

Unfortunately, we were not able to prove the conjecture in its full
generality.  So far, we are only able to prove that some of its
special cases are true.  This is not very surprising because the
conjecture touches a basic problem in circuit complexity: Can
non-linear gates help to compute linear operators? This paper is just
the first step towards this question.

\subsection{The Min-Rank Conjecture}

A \emph{completion} of a $(0,1,\ast)$-matrix $A$ is a $(0,1)$-matrix
$M$ obtained from $A$ by setting all $\ast$'s to constants $0$
and~$1$.  A \emph{canonical completion} of $A$ is obtained by setting
all $\ast$'s in $A$ to~$0$.

If $A$ is an $m$-by-$n$ matrix, then each its completion $M$ defines a
linear operator mapping each vector $\x\in\gf^n$ to a vector $M\x\in
\gf^m$. Besides such (linear) operators we also consider general
ones. Each operator $G:\gf^n\to\gf^m$ can be looked at as a sequence
$G=(g_1,\ldots,g_m)$ of $m$ boolean functions $g_i:\gf^n\to\gf$.

We say that an operator $G=(g_1,\ldots,g_m)$ is \emph{consistent} with
an $m$-by-$n$ $(0,1,\ast)$-matrix $A=(a_{ij})$ if the $i$th boolean
function $g_i$ can only depend on those variables $x_j$ for which
$a_{ij}=\ast$. That is, the $i$th component $g_i$ of $G$ can only
depend on variables on which the $i$th row of $A$ has stars (see
Example \ref{ex:1}).

\begin{dfn}\rm
  With some abuse in notation, we call a set $L\subseteq\gf^n$ a
  \emph{solution} for a partial matrix $A$ if there is a completion
  $M$ of $A$ and an operator $G$ such that $G$ is consistent with $A$
  and $M\x=G(\x)$ holds for all $\x\in L$. A~solution $L$ is
  \emph{linear} if it forms a linear subspace of $\gf^n$ over $GF_2$.
\end{dfn}

That is, a solution for $A$ is a \emph{set} $L$ of $(0,1)$-vectors of
the form $L=\{\x\colon M\x=G(\x)\}$, where $M$ is a completion of $A$,
and $G$ is an operator consistent with $A$.  A solution $L$ is linear,
if $\x\+\y\in L$ for all $\x,\y\in L$.

Since, besides the consistency, there are no other restrictions on the
operator $G$ in the definition of the solution $L$, we can always
assume that $M$ is the canonical completion of $A$ (with all stars set
to $0$).

\begin{obs}[Canonical completions]
  If $L=\{\x\colon M\x=G(\x)\}$ is a solution for $A$, and $M'$ is
  the canonical completion of $A$, then there is an operator $G'$ such
  that $G'$ is consistent with $A$ and $L=\{\x\colon M'\x=G'(\x)\}$.
\end{obs}

\begin{proof}
  The $i$th row $\m_i$ of $M$ must have the form $\m_i=\m_i'+\p_i$,
  where $\m_i'\in\gf^n$ is the $i$th row of the canonical completion
  $M'$ of $A$, and $\p_i\in\gf^n$ is a vector with no $1$'s in
  positions where the $i$th row of $A$ has no stars. We can then
  define an operator $G'=(g_1',\ldots,g_m')$ by
  $g_i'(\x):=g_i(\x)\+\skl{\p_i}{\x}$. (As customary, the scalar
  product of two vectors $\x,\y\in\gf^n$ over $GF_2$ is
  $\skl{\x}{\y}=\sum_{i=1}^nx_iy_i\bmod{2}$.)  Since $G$ was
  consistent with $A$, the new operator $G'$ is also consistent with
  $A$.  Moreover, for every vector $\x\in\gf^n$, we have that
  $\skl{\m_i}{\x}=g_i(\x)$ iff $\skl{\m_i'}{\x}=g_i'(\x)$.
\end{proof}

We are interested in how much the maximum $\opt{A}=\max_L |L|$ over
all solutions $L$ for $A$ can exceed the maximum $\lin{A}=\max_L |L|$
over all linear solutions $L$ for $A$.  It can be shown
(Corollary~\ref{cor:lin-mr} below) that
\[
\lin{A}=2^{n-\minrk{A}}\,,
\]
where $\minrk{A}$ is the \emph{min-rank} of $A$ defined as the
smallest possible rank of its completion:
\[
\minrk{A}=\min\{\rk{M}\colon\mbox{$M$ is a completion of $A$}\}\,.
\]

If we only consider \emph{constant} operators $G$, that is, operators
with $G(\x)=\b$ for some $\b\in\gf^m$ and all $\x\in\gf^n$, then
Linear Algebra tells us that no solution for $A$ can have more than
$2^{n-r}$ vectors, where $r=\rk{M}$ is the rank (over $GF_2$) of the
canonical completion $M$ of~$A$, obtained by setting all stars to~$0$.

If we only consider \emph{affine} operators $G$, that is, operators of
the form $G(\x)=H\x\+\b$ where $H$ is an $m$-by-$n$ $(0,1)$-matrix,
then no solution for $A$ can have more than $2^{n-\minrk{A}}$ vectors,
because then the consistency of $G(\x)$ with $A$ ensures that, for
every completion $M$ of $A$, the matrix $M\+H$ is a completion of $A$
as well.
\begin{rem}\label{rem:one-star}\rm
  This last observation implies, in particular, that $\opt{A}\leq
  2^{n-\minrk{A}}$ for all $(0,1,\ast)$-matrices $A$ with at most one
  $\ast$ in each row: In this case each $g_i$ can depend on at most
  one variable, and hence, must be a linear boolean function.
\end{rem}
We conjecture that a similar upper bound also holds for \emph{any}
operator $G$, as long as it is consistent with~$A$. That is, we
conjecture that linear operators are almost optimal.

\begin{conj}[Min-Rank Conjecture]\label{conj:goal}
  There exists a constant $\epsilon>0$ such that for every $m$-by-$n$
  $(0,1,\ast)$-matrix~$A$ we have that $\opt{A}\leq
  2^{n-\epsilon\cdot\minrk{A}}$ or, equivalently,
  \begin{equation}\label{eq:first}
    \opt{A}\leq 2^n\bigg(\frac{\lin{A}}{2^n} \bigg)^{\epsilon}\,.
  \end{equation}
\end{conj}

\begin{rem}\rm
To have consequences for log-depth circuits, it would be enough,
by Lemma~\ref{lem:valiant}, that the conjecture holds at least
for $\epsilon=o(1/\log\log n)$.
\end{rem}

\begin{ex}\label{ex:1}\rm
  To illustrate the introduced concepts, let us consider the following
  system of $3$ equations in~$6$ variables:
  \begin{align}
    x_1\+x_6&=x_3\cdot x_5\nonumber \\
    x_2\+x_3\+x_4&=x_1\cdot (x_5\+x_6) \label{eq:system}\\
    x_4&=(x_2\+x_5)\cdot(x_3\+x_6)\nonumber\,.
  \end{align}
  The corresponding $(0,1,\ast)$-matrix for this system is
  \begin{equation}\label{eq:matrix}
    A=
    \begin{pmatrix}
      1 & 0 & \ast & 0 & \ast & 1\\
      \ast & 1 & 1 & 1 & \ast & \ast\\
      0 & \ast & \ast & 1 & \ast & \ast
    \end{pmatrix}\,,
  \end{equation}
  and the system itself has the form $M\x=G(\x)$, where $M$ is the
  canonical completion of~$A$:
  \[
  M=
  \begin{pmatrix}
    1 & 0 & \underline{0} & 0 & \underline{0} & 1\\
    \underline{0} & 1 & 1 & 1 & \underline{0} & \underline{0}\\
    0 & \underline{0} & \underline{0} & 1 & \underline{0} &
    \underline{0}
  \end{pmatrix}\,,
  \]
  and $G=(g_1,g_2,g_3):\gf^6\to\gf^3$ is an operator with
  \begin{align*}
    g_1(\x)&= x_3\cdot x_5\,;\\
    g_2(\x)&= x_1\cdot (x_5\+x_6)\,;\\
    g_3(\x)&= (x_2\+x_5)\cdot(x_3\+x_6)\,.
  \end{align*}
  The min-rank of $A$ is equal $2$, and is achieved by the following
  completion:
  \[
  M'=
  \begin{pmatrix}
    1 & 0 & \underline{0} & 0 & \underline{0} & 1\\
    \underline{0} & 1 & 1 & 1 & \underline{0} & \underline{0}\\
    0 & \underline{1} & \underline{1} & 1 & \underline{0} &
    \underline{0}
  \end{pmatrix}\,.
  \]
\end{ex}

 \subsection{Our results}

 In Section~\ref{sec:log-depth} we prove the main consequence of the
 min-rank conjecture for boolean circuits: If true, it would imply
 that non-linear gates are powerless when computing linear operators
 $M\x$ by depth-$2$ circuits (Lemmas~\ref{lem:circ1}
 and~\ref{lem:circ2}).

 In Sections~\ref{sec:special} and \ref{sec:row-col} we prove some
 partial results supporting Conjectures~\ref{conj:gen} and
 \ref{conj:goal}. We first show (Corollary~\ref{cor:match}) that every
 depth-$2$ circuit of width $w$ computing a linear operator can be
 transformed into an equivalent \emph{linear} depth-$2$ circuit of the
 same degree and width at most $w$ plus the maximum number of wires in
 a matching formed by the input-output wires of the original circuit.

 We then prove two special cases of Min-Rank Conjecture.  A set of
 $(0,1,\ast)$-vectors is \emph{independent} if they cannot be made
 linearly dependent over $GF_2$ by setting stars to constants $0$ and
 $1$.  If $A$ is a $(0,1,\ast)$-matrix, then the upper bound
 $\opt{A}\leq 2^{n-r}$ holds if the matrix $A$ contains $r$
 independent columns (Theorem~\ref{thm:weak-col}).  The same upper
 bound also holds if $A$ contains $r$ independent rows, and the sets
 of star positions in these rows form a chain with respect to
 set-inclusion (Theorem~\ref{thm:weak-row}).

 After that we concentrate on the \emph{structure} of solutions.  In
 Section~\ref{sec:cayley} we show that solutions for a
 $(0,1,\ast)$-matrix $A$ are precisely independent sets in a Cayley
 graph over the Abelian group $(\gf^n,\+)$ generated by a special set
 $K_A\subseteq\gf^n$ of vectors defined by the matrix $A$
 (Theorem~\ref{thm:cayley}).

 In Section~\ref{sec:linear} we first show that every linear solution
 for $A$ lies in the kernel of some completion of $A$
 (Theorem~\ref{thm:lin-str}).  This, in particular, implies that
 $\lin{A}=2^{n-\minrk{A}}$ (Corollary~\ref{cor:lin-mr}), and gives an
 alternative definition of the min-rank $\minrk{A}$ as the smallest
 rank of a boolean matrix $H$ such that $H\x\neq\nulis$ for all $\x\in
 K_A$ (Corollary~\ref{cor:mr}).  In Section~\ref{sec:gen} we show that
 non-linear solutions $L$ must be ``very non-linear'': 
if $s$ is the maximum number of  $\ast$'s in a row of~$A$, and if $L$
contains a linear space $V$ such that no nozero vector with $s$ or fewer
$1$'s is orthogonal to $V$, then $L$ is contained in a \emph{linear}
solution for~$A$ (Theorem~\ref{thm:struct2}).

 In Section~\ref{sec:codes} we consider the relation of the min-rank
 conjecture with error-correcting codes.  We define
 $(0,1,\ast)$-matrices $A$, the solutions for which are
 error-correcting codes, and show that the min-rank conjecture for
 these matrices is true: In this case the conjecture is implied by
 well known lower and upper bounds on the size of linear and nonlinear
 error correcting codes (Lemma~\ref{lem:codes}).

 For readers convenience, we summarize the introduced concepts
 at the end of the paper (see Table~\ref{tab:1}).

 \section{Min-rank conjecture and depth-$2$ circuits}
 \label{sec:log-depth}

 Let $F$ be a depth-$2$ circuit computing a linear operator $\x\to
 M\x$, where $M$ is an $m$-by-$n$ $(0,1)$-matrix.  Say that the
 $(i,j)$th entry of $M$ is \emph{seen} by the circuit, if there is a
 direct wire from $x_j$ to the $i$th output gate. Replace all entries
 of $M$ seen by the circuit with $\ast$'s, and let $A_F$ be the
 resulting $(0,1,\ast)$-matrix. That is, given a depth-$2$ circuit $F$
 computing a linear operator $\x\to M\x$, we replace by $\ast$'s all
 entries of $M$ seen by the circuit, and denote the resulting
 $(0,1,\ast)$-matrix by $A_F$.  Note that the original matrix $M$ is
 one of the completions of $A_F$; hence, $\rk{M}\geq\minrk{A_F}$.

\begin{lem}\label{lem:circ}
  Every linear depth-$2$ circuit $F$ has $\width{F}\geq \minrk{A_F}$.
\end{lem}
In particular, if $F$ computes a linear operator $\x\mapsto M\x$ and
has no direct input-output wires at all, then $A_F=M$ and
$\width{F}\geq \rk{M}$.

\begin{proof}
  Let $M\x$ be a linear operator computed by $F$. Every assignment of
  constants to direct input-output wires leads to a depth-$2$ circuit
  of degree $d=0$ computing a linear operator $B\x$, where $B$ is a
  completion of $A_F$.  This operator takes $2^{\rk{B}}$ different
  values.  Hence, the operator $H:\gf^n\to\gf^w$ computed by
  $w=\width{F}$ boolean functions on the middle layer of $F$ must take
  at least so many different values, as well. This implies that the
  width $w$ must be large enough to fulfill $2^w\geq 2^{\rk{B}}$, from
  which $w\geq \rk{B}\geq \minrk{A_F}$ follows.
\end{proof}

\begin{lem}\label{lem:circ1}
  Every depth-$2$ circuit $F$ computing a linear operator can be
  transformed into an equivalent linear depth-$2$ circuit of the same
  degree and width at most $\minrk{A_F}$.
\end{lem}

Together with Lemma~\ref{lem:circ}, this implies that
$\width{F}=\minrk{A_F}$ for every optimal linear depth-$2$
circuit~$F$.

\begin{proof}
  Let $\x\to M\x$ be the operator computed by $F$, and let $A=A_F$ be
  the $(0,1,\ast)$-matrix of $F$.  We can construct the desired
  \emph{linear} depth-$2$ circuit computing $M\x$ as follows. Take a
  completion $B$ of $A$ with $\rk{B}=\minrk{A}$. By the definition of
  completions, the $i$th row $\b_i$ of $B$ has the form
  $\b_i=\a_i+\p_i$, where $\a_i$ is the $i$th row of $A$ with all
  stars set to $0$, and $\p_i$ is a $(0,1)$-vector having no $1$'s in
  positions, where this row of $A$ has non-stars.  The $i$th row
  $\m_i$ of the original $(0,1)$-matrix $M$ is of the form
  $\m_i=\a_i+\m_i'$, where $\m_i'$ is a $(0,1)$-vector which coincides
  with $\m_i$ in all positions, where the $i$th row of $A$ has stars,
  and has $0$'s elsewhere.

  The matrix $B$ has $r=\rk{B}=\minrk{A}$ linearly independent rows.
  Assume w.l.o.g. that these are the first rows $\b_1,\ldots,\b_r$ of
  $B$, and add $r$ linear gates computing the scalar products
  $\skl{\b_1}{\x},\ldots,\skl{\b_r}{\x}$ over $GF_2$ on the middle
  layer. Connect by wires each of these linear gates with all input
  and all output nodes.  Note that the $i$th output gate, knowing the
  vectors $\p_i$ and $\m_i'$, can compute both scalar products
  $\skl{\p_i}{\x}$ and $\skl{\m_i'}{\x}$ by only using existing direct
  wires from inputs $x_1,\ldots,x_n$ to this gate. Hence, using the
  $r$ linear gates $\skl{\b_1}{\x},\ldots,\skl{\b_r}{\x}$ on the
  middle layer, the $i$th output gate, for $i\leq r$, can also compute
  the whole scalar product $\skl{\m_i}{\x}$ of the input vector with
  the $i$th row of $M$ by:
  \[
  \skl{\m_i}{\x}=\skl{\a_i}{\x}\+\skl{\m_i'}{\x}
  =\skl{\b_i}{\x}\+\skl{\p_i}{\x}\+\skl{\m_i'}{\x}\,.
  \]
  For $i>r$, just replace vector $\b_i$ in this expression by the
  corresponding linear combination of $\b_1,\ldots,\b_r$.  We have
  thus constructed an equivalent linear depth-$2$ circuit of the same
  degree and of width $r=\minrk{A_F}$.
\end{proof}

By Lemma~\ref{lem:circ1}, the main question is: How much the width of
a circuit $F$ can be smaller than the min-rank of its matrix $A_F$?
Ideally, we would like to have that $\width{F}\geq \epsilon\cdot
\minrk{A_F}$: then the width of the resulting \emph{linear} circuit
would be at most $1/\epsilon$ times larger than that of the original
circuit~$F$.

Lemma~\ref{lem:circ} lower bounds the width of \emph{linear} circuits
$F$ in terms of the min-rank of their $(0,1,\ast)$-matrices $A_F$. We
now show that the Min-Rank Conjecture implies a similar fact also for
general (non-linear) circuits.

\begin{lem}\label{lem:circ2}
  For every depth-$2$ circuit $F$ computing a linear operator in $n$
  variables, we have that
  \[
  \width{F}\geq n-\log_2\opt{A_F}\,.
  \]
\end{lem}
Hence, the Min-Rank Conjecture (stating that $\opt{A}\leq
2^{n-\epsilon\cdot\minrk{A}}$) implies that $\width{F}\geq
\epsilon\cdot \minrk{A_F}$.
\begin{proof}
  Let $M$ be an $m$-by-$n$ $(0,1)$-matrix.  Take a depth-$2$ circuit
  $F$ of width $w$ computing $M\x$, and let $A_F$ be the corresponding
  $(0,1,\ast)$-matrix. Let $H=(h_1,\ldots,h_w)$ be an operator
  computed at the gates on the middle layer, and $G=(g_1,\ldots,g_m)$
  an operator computed at the gates on the output layer. Hence,
  $M\x=G(\x,H(\x))$ for all $\x\in\gf^n$.  Fix a vector $\b\in\gf^w$
  for which the set $L=\{\x\in\gf^n\colon M\x=G(\x,\b)\}$ is the
  largest one; hence, $|L|\geq 2^{n-w}$.  Note that the operator
  $G'(\x):=G(\x,\b)$ must be consistent with $A$: its $i$th component
  $g_i'(\x)$ can only depend on input variables $x_j$ to which the
  $i$th output gate $g_i$ is connected. Hence, $L$ is a solution for
  $A_F$, implying that $\opt{A_F}\geq |L|\geq 2^{n-w}$ from which the
  desired lower bound $w\geq n-\log_2\opt{A_F}$ on the width of $F$
  follows.
\end{proof}

We can now show that the Min-Rank Conjecture
(Conjecture~\ref{conj:goal}) indeed implies the Linearization
Conjecture (Conjecture~\ref{conj:gen}).

\begin{cor}\label{cor:circ}
  Conjecture~\ref{conj:goal} implies Conjecture~\ref{conj:gen}.
\end{cor}

\begin{proof}
  Let $F$ be a depth-$2$ circuit computing a linear operator in $n$
  variables.  Assuming Conjecture~\ref{conj:goal},
  Lemma~\ref{lem:circ2} implies that $\epsilon\cdot \minrk{A_F}\leq
  n-\log_2\opt{A_F}\leq \width{F}$.  By Lemma~\ref{lem:circ1}, the
  circuit $F$ can be transformed into an equivalent linear depth-$2$
  circuit of the same degree and width at most $\minrk{A_F}\leq
  \width{F}/\epsilon$.
\end{proof}
Hence, together with Valiant's result, the Min-Rank Conjecture implies
that a linear operator $M\x$ requires a super-linear number of gates
in any log-depth circuit over $\{\&,\lor,\neg\}$, if every depth-$2$
circuit for $M\x$ over $\{\oplus,1\}$ of width $w=O(n/\ln\ln n)$
requires degree~$n^{\Omega(1)}$.

Finally, let us show that the only ``sorrow'', when trying to
linearize a depth-$2$ circuit, is the possible non-linearity of
\emph{output} gates---non-linearity of gates on the middle layer is no
problem.

\begin{lem}\label{lem:circ3}
  Let $F$ be a depth-$2$ circuit computing a linear operator.  If all
  gates on the output layer are linear boolean functions, then $F$ can
  be transformed into an equivalent linear depth-$2$ circuit of the
  same degree and width.
\end{lem}

\begin{proof}
  Let $M$ be an $m$-by-$n$ $(0,1)$-matrix, and let $F$ be a depth-$2$
  circuit of width $w$ computing $M\x$. Let $H=(h_1,\ldots,h_w)$ be
  the operator $H:\gf^n\to\gf^w$ computed by the gates on the middle
  layer.  Assume that all output gates of $F$ are linear boolean
  functions.  Let $B$ be the $m$-by-$n$ adjacency $(0,1)$-matrix of
  the bipartite graph formed by the direct input-output wires, and $C$
  be the $m$-by-$w$ adjacency $(0,1)$-matrix of the bipartite graph
  formed by the wires joining the gates on the middle layer with those
  on the output layer.  Then
  \[
  M\x=B\x\+C\cdot H(\x)\qquad\mbox{for all $\x\in\gf^n$,}
  \]
  where $C\cdot H(\x)$ is the product of the matrix $C$ with the
  vector $\y=H(\x)$.  Hence,
  \begin{equation}\label{eq:middle1}
    C\cdot H(\x)=D\x
  \end{equation}
  is a linear operator with $D=M\+B$.  Write each vector
  $\x=(x_1,\ldots,x_n)$ as the linear combination
  \begin{equation}\label{eq:middle2}
    \x=\sum_{i=1}^n x_i\e_i
  \end{equation}
  of unit vectors $\e_1,\ldots,\e_n\in\gf^n$, and replace the operator
  $H$ computed on the middle layer by a \emph{linear} operator
  \begin{equation}\label{eq:middle3}
    H'(\x):=\sum_{i=1}^n x_iH(\e_i) \pmod{2}\,.
  \end{equation}
  Then, using the linearity of the matrix-vector product, we obtain
  that (with all sums mod~$2$):
  \begin{xalignat*}{2}
    C\cdot H(\x)&= D\cdot \Big(\sum x_i\e_i\Big)
    && \mbox{by (\ref{eq:middle1}) and (\ref{eq:middle2})}\\
    &=\sum x_i D\e_i && \mbox{linearity}\\
    &=\sum x_i C\cdot H(\e_i) &&  \mbox{by (\ref{eq:middle1})}\\
    &=C\cdot \Big(\sum x_i H(\e_i)\Big) && \mbox{linearity}\\
    &=C\cdot H'(\x) && \mbox{by (\ref{eq:middle3})}\,.
  \end{xalignat*}
  Hence, we again have that $M\x=B\x\+C\cdot H'(\x)$, meaning that the
  obtained \emph{linear} circuit computes the same linear operator
  $M\x$.
\end{proof}

\section{Bounds on $\opt{A}$}
\label{sec:special}

Recall that $\opt{A}$ is the largest possible number of vectors in a
solution for a given $(0,1,\ast)$-matrix~$A$. The simplest properties
of this parameter are summarized in the following

\begin{lem}\label{lem:split}
  Let $A$ be an $m$-by-$n$ $(0,1,\ast)$-matrix.  If $A'$ is obtained
  by removing some rows of $A$, then $\opt{A'}\geq\opt{A}$.  If
  $A=[B,C]$ where $B$ is an $m$-by-$p$ submatrix of $A$ for some
  $1\leq p\leq n$, then
  \[
  \opt{B}\cdot\opt{C}\leq \opt{A}\leq \opt{B}\cdot 2^{n-p}\,.
  \]
\end{lem}

\begin{proof}
  The first claim $\opt{A'}\geq \opt{A}$ is obvious, since addition of
  new equations can only decrease the number of solutions in any
  system of equations.

  To prove $\opt{A}\leq \opt{B}\cdot 2^{n-q}$, take an optimal
  solution $L_A=\{\x\colon M\x=G(\x)\}$ for $A$; hence,
  $|L_A|=\opt{A}$.  Fix a vector $\b\in\gf^{n-p}$ for which the set
  \[
  L_B=\{\y\in \gf^p\colon (\y,\b)\in L_A\}
  \]
  is the largest one; hence, $|L_B|\geq \opt{A}/2^{n-p}$. The
  completion $M$ of $A$ has the form $M=[M',M'']$, where $M'$ is a
  completion of $B$ and $M''$ is a completion of $C$. If we define an
  operator $G':\gf^p\to\gf^m$ by
  \[
  G'(\y):=G(\y,\b)\+M''\b\,,
  \]
  then $M'\y=G'(\y)$ for all $\y\in L_B$.  Hence, $L_B$ is a solution
  for $B$, implying that $\opt{A}\leq |L_B|\cdot 2^{n-p} \leq
  \opt{B}\cdot 2^{n-p}$.

  To prove $\opt{A}\geq \opt{B}\cdot\opt{C}$, let
  $L_B=\{\y\in\gf^p\colon M'\y=G'(\y)\}$ be an optimal solution for
  $B$, and let $L_C=\{\z\in\gf^{n-p}\colon M''\z=G''(\z)\}$ be an
  optimal solution for $C$.  For any pair $\x=(\y,\z)\in L_B\times
  L_C$, we have that $M\x=G(\x)$, where $M=[M',M'']$ and
  $G(\y,\z):=G'(\y)\+G''(\z)$.  Hence, the set $L_B\times
  L_C\subseteq\gf^n$ is a solution for $A$, implying that
  $\opt{B}\cdot\opt{C}=|L_B\times L_C|\leq \opt{A}$, as claimed.
\end{proof}

Let $A$ be an $m$-by-$n$ $(0,1,\ast)$-matrix.  The min-rank conjecture
claims that the largest number $\opt{A}$ of vectors in a solution for
$A$ can be upper bounded in terms of the min-rank of $A$ as
$\opt{A}\leq 2^{n-\epsilon\cdot\minrk{A}}$.  The claim is true if the
min-rank of $A$ is ``witnessed'' by some $(0,1)$-submatrix of $A$,
that is, if $A$ contains a $(0,1)$-submatrix of rank equal to the
min-rank of~$A$. This is a direct consequence of the following simple

\begin{lem}\label{lem:star-free}
  If $A$ is an $m$-by-$n$ $(0,1,\ast)$-matrix, then $\opt{A}\leq
  2^{n-\rk{B}}$ for every $(0,1)$-submatrix $B$ of $A$.
\end{lem}

\begin{proof}
  Let $B$ be a $p$-by-$q$ $(0,1)$-submatrix of $A$.  Since $B$ has no
  stars, only constant operators can be consistent with $B$. Hence, if
  $L\subseteq\gf^q$ is a solution for $B$, then there must be a vector
  $\b\in\gf^p$ such that $B\x=\b$ for all $\x\in L$.  This implies
  $|L|\leq 2^{q-\rk{B}}$. Together with Lemma~\ref{lem:split}, this
  yields $\opt{A}\leq 2^{q-\rk{B}}\cdot 2^{n-q}=2^{n-\rk{B}}$.
\end{proof}

The \emph{max-rank} $\maxrk{A}$ of a $(0,1,\ast)$-matrix $A$ is a
maximal possible rank of its completion. A \emph{line} of $A$ is
either its row or its column. A \emph{cover} of $A$ is a set $X$ of
its lines covering all stars.  Let $\cov{A}$ denote the smallest
possible number of lines in a cover of~$A$.

\begin{lem}\label{lem:maxrk}
  For every $m$-by-$n$ $(0,1,\ast)$-matrix $A$, we have that
  \[
  \opt{A}\leq 2^{n-\maxrk{A}+\cov{A}}\,.
  \]
\end{lem}

\begin{proof}
  Given a cover $X$ of the stars in $A$ by lines, remove all these
  lines, and let $A_X$ be the resulting $(0,1)$-submatrix of $A$.
  Clearly, we have: $\maxrk{A}\leq \rk{A_X}+|X|$. (In fact, it is
  shown in \cite{CJRW89} that $\maxrk{A}=\min_X
  \left(\rk{A_X}+|X|\right)$, where the minimum is over all covers $X$
  of $A$.) Take a cover $X$ of $A$ of size $|X|=\cov{A}$. Hence,
  $\maxrk{A}\leq \rk{A_X}+\cov{A}$.  Since $A_X$ is a
  $(0,1)$-submatrix of $A$, Lemma~\ref{lem:star-free} yields
  $\opt{A}\leq 2^{n-\rk{A_X}}$, where $\rk{A_X}\geq \maxrk{A}-|X|=
  \maxrk{A}-\cov{A}$.
\end{proof}

Given a depth-$2$ circuit $F$, let $\match{F}$ denote the largest
number of wires in a matching formed by direct input-output
wires. That is, $\match{F}$ is the largest number of $\ast$-entries in
the matrix $A_F$ of $F$, no two on the same line.  By the well-known
K\"onig--Egev\'ary theorem, stating  that the size of a largest
matching in a bipartite graph is equal to the smallest set of vertices
which together touch every edge, we have that
$\match{A}=\cov{A_F}$. This leads to the following

\begin{cor}\label{cor:match}
  Every depth-$2$ circuit $F$ computing a linear operator can be
  transformed into an equivalent linear depth-$2$ circuit $F'$ of the
  same degree and
  \[
  \width{F'}\leq \width{F}+\match{F}\,.
  \]
\end{cor}

\begin{proof}
  Let $A_F$ be the $(0,1,\ast)$-matrix of $F$. By
  Lemmas~\ref{lem:circ2} and \ref{lem:maxrk}, we have that
  \begin{align*}
  \width{F}&\geq n-\log_2\opt{A_F}\geq
  n-\left[n-\maxrk{A_F}+\cov{A_F}\right]\\
  & = \maxrk{A_F}-\cov{A_F} =
  \maxrk{A_F}-\match{F}\,.
  \end{align*}
  By Lemma~\ref{lem:circ1}, the circuit $F$ can be transformed into an
  equivalent linear depth-$2$ circuit of the same degree and width at
  most $\minrk{A_F}\leq \maxrk{A_F}\leq \width{F}+\match{F}$.
\end{proof}

\section{Row and column min-rank}
\label{sec:row-col}

We are now going to show that the min-rank conjecture holds for
stronger versions of min-rank---row min-rank and column min-rank.

If $A$ is a $(0,1,\ast)$-matrix of min-rank $r$ then, for every
assignment of constants to stars, the resulting $(0,1)$-matrix will
have $r$ linearly independent columns as well as $r$ linearly
independent rows.  However, for different assignments these
columns/rows may be different.  It is natural to ask whether the
min-rank conjecture is true if the matrix $A$ has $r$ columns (or $r$
rows) that remain linearly independent under any assignment of
constants to stars?

Namely, say that $(0,1,\ast)$-vectors are \emph{dependent} if they can
be made linearly dependent over $GF_2$ by setting their $\ast$-entries to
a constants $0$ and $1$; otherwise, the vectors are \emph{independent}.

\begin{rem}\rm
  The dependence of $(0,1,\ast)$-vectors can be defined by adding to
  $\gf$ a new element $\ast$ satisfying
  $\alpha\+\ast=\ast\+\alpha=\ast$ for $\alpha\in\{0,1,\ast\}$.  Then
  a set of $(0,1,\ast)$-vectors is dependent iff some its subset sums
  up to a $(0,\ast)$-vector. Indeed, if \emph{some} subset sums up to
a  $(0,\ast)$-vector, then we can set the $\ast$-entries to constants so that
the corresponding subset of $(0,1)$-vectors will sum up (over $GF_2$) to
an all-$0$ vector. On the other hand, if \emph{no} subset sums up to a
 $(0,\ast)$-vector, for every subset, there must be a position in which
all vectors in this subset have no stars, and the sum of these positions
over $GF_2$ is~$1$.
\end{rem}

\begin{rem}\rm
  A basic fact of Linear Algebra, leading to the Gauss-Algorithm, is
  that linear independence of vectors $\x,\y\in\gf^n$ implies that the
  vectors $\x+\y$ and $\y$ are linear independent as well. For
  $(0,1,\ast$)-vectors this does not hold anymore. Take, for example,
  $\x=(0,1)$ and $\y=(1,\ast)$.  Then $\x\+\y=(1,\ast)=\y$.
\end{rem}

For a $(0,1,\ast)$-matrix $A$, define its \emph{column min-rank},
$\wminrkcol{A}$, as the maximum number of independent columns, and its
\emph{row min-rank}, $\wminrkrow{A}$, as the maximum number of
independent rows. In particular, both $\wminrkrow{A}$ and
$\wminrkcol{A}$ are at least $r$ if $A$ contains an $r\times r$
``triangular'' submatrix, that is, a submatrix with zeroes below (or
above) the diagonal and ones on the diagonal:
\[
\Delta=
\begin{pmatrix}
  1 & \circledast & \circledast  & \circledast\\
  0 &     1       &   \circledast     &   \circledast         \\
  0 &     0       &       1   &    \circledast        \\
  0 & 0 & 0 & 1
\end{pmatrix}\,,
\]
where $\circledast\in\{0,1,\ast\}$.  It is clear that neither
$\wminrkcol{A}$ nor $\wminrkrow{A}$ can exceed the min-rank of
$A$. Later (Lemma~\ref{lem:gap} below) we will give an example of a
matrix $A$ where both $\wminrkcol{A}$ and $\wminrkrow{A}$ are by a
logarithmic factor smaller than $\minrk{A}$. The question about a more
precise relation between these parameters remains open (see
Problem~\ref{probl:gap}).

Albeit for $(0,1)$-matrices we always have that their row-rank
coincides with column-rank, for $(0,1,\ast)$-matrices this is no more
true.  In particular, for some $(0,1,\ast)$-matrices $A$, we have that
$\wminrkrow{A}\neq \wminrkcol{A}$.

\begin{ex}\rm
  Consider the following $(0,1,\ast)$-matrix:
  \[
  A=
  \begin{pmatrix}
    1 & 1 & * & 1 \\
    1 & 0 & 1 & *\\
    1 & * & 0 & 0
  \end{pmatrix}\,.
  \]
  Then $\wminrkrow{A}=\minrk{A}=3$ but $\wminrkcol{A}=2$.  To see that
  $\wminrkrow{A}=3$, just observe that the rows cannot be made
  linearly dependent by setting the stars to $0$ or $1$: the sum of
  all three vectors is not a $\{0,\ast\}$-vector because of the $1$st
  column, and the pairwise sums are not $\{0,\ast\}$-vectors because,
  for each pair of rows there is a column containing $0$ and $1$.  To
  see that $\wminrkcol{A}=2$, observe that the last three columns are
  dependent (each row has a star). Moreover, for every pair of these
  columns, there is an assignment of constants to stars such that
  either the resulting $(0,1)$-columns are equal or their sum equals
  the first column.
\end{ex}

We first show that the min-rank conjecture holds with ``min-rank''
replaced by ``column min-rank''.

\begin{thm}[Column min-rank]\label{thm:weak-col}
  Let $A$ be a $(0,1,\ast)$-matrix with $n$ columns and of column
  min-rank $r$. Then $\opt{A}\leq 2^{n-r}$.
\end{thm}

\begin{proof}
  Any $m$-by-$n$ $(0,1,\ast)$-matrix $B$ of column min-rank $r$ must
  contain an $m\times r$ submatrix $A$ of min-rank $r$. Since
  $\opt{B}\leq \opt{A}\cdot 2^{n-r}$ (Lemma~\ref{lem:split}), it is
  enough to show that $\opt{A}\leq 1$ for all $m$-by-$r$
  $(0,1,\ast)$-matrices $A$ of min-rank $r$.

  To do this, let $L$ be a solution for $A$. Then there is an operator
  $G=(g_1,\ldots,g_m):\gf^r\to\gf^m$ such that $G$ is consistent with
  $A$ and $\skl{\a_i}{\x}=g_i(\x)$ holds for all $\x\in L$ and all
  $i=1,\ldots,m$.  Here $\a_1,\ldots,\a_m$ are the rows of $A$ with
  all stars set to $0$.

  For the sake of contradiction, assume that $|L|\geq 2$ and fix any
  two vectors $\x\neq\y\in L$.  Our goal is to construct a vector
  $\c\in\gf^m$ and a completion $M$ of $A$ such that $M\x=M\y=\c$.
  Since $M$ must have rank $r$, this will give the desired
  contradiction, because at most $2^{r-\rk{M}}=2^0=1$ vectors $\z$ can
  satisfy $M\z=\c$.

  If $M$ is a completion of $A=(a_{ij})$, then its $i$th row must have
  the form $\m_i=\a_i\+\p_i$ where $\p_i\in\gf^n$ is some vector with
  no $1$'s in positions where the $i$th row of $A$ has no stars. To
  construct the desired vector $\p_i$ for each $i\in[m]$, we consider
  two possible cases.  (Recall that the vectors $\x$ and $\y$ are
  fixed.)

  {\bf Case 1}: $\skl{\a_i}{\x}=\skl{\a_i}{\y}$.  In this case we can
  take $\p_i=\nulis$ and $c_i=\skl{\a_i}{\x}$.  Then
  $\skl{\m_i}{\x}=\skl{\m_i}{\y}=\skl{\a_i}{\x}=c_i$, as desired.

  {\bf Case 2}: $\skl{\a_i}{\x}\neq \skl{\a_i}{\y}$.  In this case we
  have that $g_i(\x)\neq g_i(\y)$, that is, the vectors $\x$ and $\y$
  must differ in some position $j$ where the $i$th row of $A$ has a
  star. Then we can take $\p_i:=\e_j$ (the $j$th unit vector) and
  $c_i:=\skl{\a_i}{\x}\+ x_j$. With this choice of $\p_i$, we again
  have
  \[
  \skl{\m_i}{\x}=\skl{\a_i}{\x}\+ \skl{\p_i}{\x} =\skl{\a_i}{\x}\+
  \skl{\e_j}{\x}= \skl{\a_i}{\x}\+x_j=c_i
  \]
  and, since $\skl{\a_i}{\x}\neq \skl{\a_i}{\y}$ and $x_j\neq y_j$,
  \[
  \skl{\m_i}{\y}=\skl{\a_i}{\y}\+ \skl{\p_i}{\y} =\skl{\a_i}{\y}\+
  \skl{\e_j}{\y} =\skl{\a_i}{\x}\+x_j=c_i\,. \qedhere
  \]
\end{proof}

\begin{ex}\rm
  It is not difficult to verify that, for the $(0,1,\ast)$-matrix $A$
  given by (\ref{eq:matrix}), we have that
  $\wminrkcol{A}=\minrk{A}=2$. Hence, no linear solution of the system
  of semi-linear equations (\ref{eq:system}) can have more than
  $\lin{A}=2^{6-2}=32$ vectors. Theorem~\ref{thm:weak-col} implies
  that, in fact, \emph{no} solution can have more than this number of
  vectors.
\end{ex}

The situation with \emph{row} min-rank is more complicated.  In this
case we are only able to prove an upper bound $\opt{A}\leq 2^{n-r}$
under an additional restriction that the star-positions in the rows of
$A$ form a chain under set-inclusion.

Recall that $(0,1,\ast)$-vectors are \emph{independent} if they cannot
be made linearly dependent over $GF_2$ by setting stars to
constants. The row min-rank of a $(0,1,\ast)$-matrix is the largest
number $r$ of its independent rows.  Since adding new rows can only
decrease $\opt{A}$, it is enough to consider $r$-by-$n$
$(0,1,\ast)$-matrices $A$ with $\minrk{A}=r$.

If $r=1$, that is, if $A$ consists of just one row, then $\opt{A}\leq
2^{n-1}=2^{n-r}$ holds. Indeed, since $\minrk{A}=1$, this row cannot
be a $(0,\ast)$-row. So, there must be at least one $1$ in, say, the
$1$st position.  Let $L_A=\{\x\colon \skl{\a_1}{\x}=g_1(\x)\}$ be a
solution for $A$, where $\a_1$ is the row of $A$ with all stars set
to~$0$.  Take the unit vector $\e_1=(1,0,\ldots,0)$ and split the
vectors in $\gf^n$ into $2^{n-1}$ pairs $\{\x,\x\+\e_1\}$.  Since the
boolean function $g_1$ cannot depend on the first variable $x_1$, we
have that $g_1(\x\+\e_1)=g_1(\x)$. But
$\skl{\a_i}{\x\+\e_1}=\skl{\a_i}{\x}\+1\neq \skl{\a_i}{\x}$.  Hence,
at most one of the two vectors $\x$ and $\x\+\e_1$ from each pair
$\{\x,\x\+\e_1\}$ can lie in $L_A$, implying that $|L_A|\leq 2^{n-1}$.

To extend this argument for matrices with more rows, we need the
following definition.  Let $A=(a_{ij})$ be an $r$-by-$n$
$(0,1,\ast)$-matrix, and $\a_1,\ldots,\a_r$ be the rows of $A$ with
all stars set to $0$.  Let $S_i=\{j\colon a_{ij}=\ast\}$ be the set of
star-positions in the $i$th row of~$A$.  It will be convenient to
describe the star-positions by diagonal matrices.  Namely, let $D_i$
be the incidence matrix of stars in the $i$th row of $A$. That is,
$D_i$ is a diagonal $n$-by-$n$ $(0,1)$-matrix whose $j$th diagonal
entry is $1$ iff $j\in S_i$. In particular, $D_i\x=\nulis$ means that
$x_j=0$ for all $j\in S_i$.

\begin{dfn}\label{dfn:isol}\rm
  A matrix $A$ is \emph{isolated} if there exist vectors
  $\z_1,\ldots,\z_r\in\gf^n$ such that, for all $1\leq i\leq r$, we
  have $D_i\z_i=\nulis$ and
  \[
  \skl{\a_j}{\z_i}=\begin{cases}
    1 & \mbox{ if $j=i$;}\\
    0 & \mbox{ if $j<i$.}
  \end{cases}
  \]
  If $D_1\z_i=\ldots=D_i\z_i=\nulis$, then the matrix is
  \emph{strongly isolated}.
\end{dfn}

\begin{lem}\label{lem:isol}
  If $A$ is a strongly isolated $r$-by-$n$ $(0,1,\ast)$-matrix, then
  $\opt{A}\leq 2^{n-r}$.
\end{lem}

\begin{proof}
  Let $\a_1,\ldots,\a_r$ be the rows of $A$ with all stars set to
  $0$. We prove the lemma by induction on $r$.  The basis case $r=1$
  is already proved above. For the induction step $r-1\mapsto r$, let
  \[
  L_A=\{\x\in\gf^n\colon \skl{\a_i}{\x}=g_i(\x) \mbox{ for all }
  i=1,\ldots,r\}
  \]
  be an optimal solution for $A$, and let $B$ be a submatrix of $A$
  consisting of its first $r-1$ rows. Then
  \[
  L_B=\{\x\in\gf^n\colon \skl{\a_i}{\x}=g_i(\x) \mbox{ for all }
  i=1,\ldots,r-1\}
  \]
  is a solution for~$B$.  Since $A$ is strongly isolated, the matrix
  $B$ is strongly isolated as well.  The induction hypothesis implies
  that $|L_B|\leq 2^{n-(r-1)}$.

  Let $\z=\z_r$ be the $r$-th isolating vector.  For each row
  $i=1,\ldots,r-1$, the conditions $\skl{\z}{\a_i}=0$ and
  $D_i\z=\nulis$ imply that $\skl{(\x\+\z)}{\a_i}=\skl{\x}{\a_i}$ and
  $g_i(\x\+\z)=g_i(\x)$.  That is,
  \[
  \mbox{$\x\in L_B$ iff $\x\+\z\in L_B$.}
  \]
  For the $r$th row, the conditions $\skl{\z}{\a_r}=1$ and
  $D_r\z=\nulis$ imply that $\skl{(\x\+\z)}{\a_r}\neq \skl{\x}{\a_r}$
  whereas $g_r(\x\+\z)=g_r(\x)$.  That is,
  \[
  \mbox{$\x\in L_A$ iff $\x\+\z\not\in L_A$.}
  \]
  Hence, for every vector $\x\in L_B$, only one of the vectors $\x$
  and $\x\+\z$ can belong to $L_A$, implying that
\[
\opt{A}=|L_A|\leq |L_B|/2\leq 2^{n-r}\,. \qedhere
\]
\end{proof}

We are now going to show that $(0,1,\ast)$-matrices with some
conditions on the distribution of stars in them are strongly isolated.
For this, we need the following two facts.  A \emph{projection} of a
vector $\x=(x_1,\ldots,x_n)$ onto a set of positions
$I=\{i_1,\ldots,i_k\}$ is the vector
\[
\proj{\x}{I}=(x_{i_1},\ldots,x_{i_k})\,.
\]
A $(0,1,\ast)$-vector $\x$ is \emph{independent of}
$(0,1,\ast)$-vectors $\y_1,\ldots,\y_k$ if no completion of $\x$ can
be written as a linear combination of some completions of these
vectors.

\begin{lem}\label{lem:proj}
  Let $\x,\y_1,\ldots,\y_k$ be $(0,1,\ast)$-vectors, and $I=\{i\colon
  x_i\neq \ast\}$.  If $\x$ is independent of $\y_1,\ldots,\y_k$, then
  $\proj{\x}{I}$ is also independent of
  $\proj{\y_1}{I},\ldots,\proj{\y_k}{I}$.
\end{lem}

\begin{proof}
  Assume that $\proj{\x}{I}$ is dependent on the projections
  $\proj{\y_1}{I},\ldots,\proj{\y_k}{I}$.  Then there is an assignment
  of stars to constants in the vectors $\y_i$ such that $\proj{\x}{I}$
  can be written as a linear combination of the projections
  $\proj{\y_1'}{I},\ldots,\proj{\y_k'}{I}$ on $I$ of the resulting
  $(0,1)$-vectors $\y_1',\ldots,\y_k'$. But since $\x$ has stars in
  all positions outside $I$, these stars can be set to appropriate
  constants so that the resulting $(0,1)$-vector $\x'$ will be a
  linear combination of $\y_1',\ldots,\y_k'$, a contradiction.
\end{proof}

\begin{lem}\label{lem:lin}
  Let $\a\in\gf^n$ be a vector and $M$ be an $m$-by-$n$ $(0,1)$-matrix
  of rank $r\leq n-1$. If $\a$ is linearly independent of the rows of
  $M$, then there exists a set $Z\subseteq \gf^n$ of $|Z|\geq
  2^{n-r-1}$ vectors such that, for all $\z\in Z$, we have
  $\skl{\z}{\a}=1$ and $M\z=\nulis$.
\end{lem}

\begin{proof}
  Let $Z=\{\z\colon M\z=\nulis, \skl{\a}{\z}=1\}$, and let $M'$ be the
  matrix $M$ with an additional row $\a$.  Note that
  $Z=\ker{M}\setminus \ker{M'}$, where $\ker{M}=\{\z\colon
  M\z=\nulis\}$ is the kernel of $M$.  Since $\rk{M'}=\rk{M}+1\leq n$,
  we have that $|\ker{M'}|=|\ker{M}|/2$, implying that
  \[
  |Z|=|\ker{M}\setminus \ker{M'}|=|\ker{M}|/2\geq
  2^{n-r-1}\,. \qedhere
  \]
\end{proof}

\begin{lem}\label{lem:almost}
  If $A$ is an $r$-by-$n$ $(0,1,\ast)$-matrix with $\minrk{A}=r$, then
  $A$ is isolated.
\end{lem}

\begin{proof} Let $\a_1,\ldots,\a_r$ be the rows of $A$ with all stars
  set to $0$.  Let $I\subseteq\{1,\ldots,n\}$ be the set of all
  star-free positions in the $i$th row of $A$, and consider an
  $(r-1)$-by-$|I|$ $(0,1)$-matrix $M_i$ whose rows are the projections
  $\a_j'=\proj{\a_j}{I}$ of vectors $\a_j$ with $j\neq i$ onto the set
  $I$.  By Lemma~\ref{lem:proj}, the projection $\a_i'=\proj{\a_i}{I}$
  of the $i$th vector $\a_i$ onto $I$ cannot be written as a linear
  combination of the rows of $M_i$; hence, $\rk{M_i}\leq |I|-1$.
  Since $2^{|I|-\rk{M_i}-1}\geq 2^0=1$, Lemma~\ref{lem:lin} gives us a
  vector $\z_i'\in\gf^{|I|}$ such that $\skl{\z_i'}{\a_i'}=1$ and
  $\skl{\z_i'}{\a_j'}=0$ for all $j\neq i$. But then
  $\z_i:=(\z_i',\nulis)$ is the desired $(0,1)$-vector:
  $D_i\z_i=D_i\cdot \nulis=\nulis$,
  $\skl{\z_i}{\a_i}=\skl{\z_i'}{\a_i'}=1$, and
  $\skl{\z_i}{\a_j}=\skl{\z_i'}{\a_j'}=0$ for all rows $j\neq i$.
\end{proof}

Say that an $r$-by-$n$ $(0,1,\ast)$-matrix $A$ is \emph{star-monotone}
if the sets $S_1,\ldots,S_r$ of star-positions in its rows form a
chain, that is, if $S_1\subseteq S_2\subseteq\ldots\subseteq S_r$.

\begin{thm}[Star-monotone matrices]\label{thm:weak-row}
  Let $A$ be a $(0,1,\ast)$-matrix with $n$ columns. If $A$ contains
  an $r$-by-$n$ star-monotone submatrix of min-rank $r$, then
  $\opt{A}\leq 2^{n-r}$.
\end{thm}

\begin{proof} Since addition of new rows can only decrease the size of
  a solution, we can assume that $A$ itself is an $r$-by-$n$
  star-monotone matrix of min-rank $r$.  Let $\a_1,\ldots,\a_r$ be the
  rows of $A$ with all stars set to $0$.  By Lemma~\ref{lem:almost},
  the matrix $A$ is isolated.  That is, there exist vectors
  $\z_1,\ldots,\z_r\in\gf^n$ such that: $\skl{\a_i}{\z_j}=1$ iff
  $i=j$, and $D_i\z_i=\nulis$ for all $1\leq i\leq r$. Since
  $S_j\subseteq S_i$ for all $j<i$, this last condition implies that
  $D_j\z_i=\nulis$ for all $1\leq j<i\leq r$, that is, $A$ is strongly
  isolated.  Hence, we can apply Lemma~\ref{lem:isol}.
\end{proof}

\section{Solutions as independent sets in Cayley graphs}
\label{sec:cayley}

Let $A=(a_{ij})$ be an $m$-by-$n$ $(0,1,\ast)$-matrix.  In the
definition of solutions $L$ for $A$ we take a completion $M$ of $A$
and an operator $G(\x)$, and require that $M\x=G(\x)$ for all $\x\in
L$. The operator $G=(g_1,\ldots,g_m)$ can be arbitrary---the only
restriction is that its $i$th component $g_i$ can only depend on
variables corresponding to stars in the $i$th row of~$A$.  In this
section we show that the actual \emph{form} of operators $G$ can be
ignored---only star-positions are important.  To do this, we associate
with $A$ the following set of ``forbidden'' vectors:
\[
K_A=\{\x\in \gf^n\colon\mbox{$\exists i\in[m]$\ \ $D_i\x=\nulis$ and
  $\skl{\a_i}{\x}=1$}\}\,,
\]
where $D_i$ is the incidence $n$-by-$n$ $(0,1)$-matrix of stars in the
$i$th row of $A$, and $\a_i$ is the $i$th row of $A$ with all stars
set to $0$.  Hence, $K_A$ is a union $K_A=\bigcup_{i=1}^m K_i$ of $m$
affine spaces
\[
K_i=\bigg\{\x\colon \begin{pmatrix}D_i\\ \a_i\end{pmatrix}\x
=\begin{pmatrix} \nulis \\ 1\end{pmatrix}\bigg\}\,.
\]

\begin{lem}\label{lem:forb}
For every vector $\x\in\{0,1\}^n$, $\x\in K_A$ if and only if $M\x\neq\nulis$
for all completions $M$ of $A$.
\end{lem}

\begin{proof}
$(\Rightarrow)$: Take a vector $\x\in K_A$. Then there exists an $i\in[m]$
such that vector $\x$ has zeroes in all positions, where the $i$th row
of $A$ has stars, and   $\skl{\a_i}{\x}=1$, where $\a_i$
is obtained by setting all stars in this row to $0$. So, if $\b_i$ is any
completion of the $i$th row of $A$ then $\skl{\b_i}{\x}
=\skl{\a_i}{\x}=1$. Thus, the scalar product of $\x$ with the $i$th row
of any completion of $A$ must be equal to~$1$.

$(\Leftarrow)$: Take a vector $\x\not\in K_A$. We have to show that then
$M\x=\nulis$ for at least one completion $M$ of $A$. The fact that $\x$
does not belong to $K_A$ means that for each $i\in[m]$ either (i)
$\skl{\a_i}{\x}=0$, or (ii) $\skl{\a_i}{\x}=1$ but vector $\x$ has a $1$
in some position $j$, where the $i$th row of $A$ has a star.
We can therefore construct the $i$th row $\m_i$ of the desired
completion $M$ of $A$ with $M\x=\nulis$
by taking $\m_i=\a_i$, if (i), and $m_i=\a_i+\e_j$, if (ii). In both cases
we have  $\skl{\m_i}{\x}=0$, as desired.
\end{proof}

The \emph{sum-set} of two sets of vectors $S,T\subseteq\gf^n$
is the set of vectors
\[
S+T=\{\x\+\y\colon\x\in S \mbox{ and } \y\in T\}\,.
\]

\begin{thm}\label{thm:cayley}
  A set $L\subseteq\gf^n$ is a solution for $A$ if and only if
  $(L+L)\cap K_A=\emptyset$.
\end{thm}

\begin{proof}
  Observe that the sum $\x\+\y$ of two vectors belongs to $K_A$ iff
  these vectors coincide on all stars of at least one row of $A$ such
  that $\skl{\a_i}{\x}\neq \skl{\a_i}{\y}$. By this observation, we
  see that the condition $(L+L)\cap K_A=\emptyset$ is equivalent to:
  \begin{equation}\label{eq:sol0}
    \forall\x,\y\in L\ \ \forall i\in[m]:\  D_i\x= D_i\y\ \ \mbox{implies}\ \
    \skl{\a_i}{\x}= \skl{\a_i}{\y}.
  \end{equation}
  Having made this observation, we now turn to the actual proof of
  Theorem~\ref{thm:cayley}.

  $(\Rightarrow)$ Let $L$ be a solution for $A$. Hence, there is an
  operator $G=(g_1,\ldots,g_m)$ consistent with $A$ such that
  $\skl{\a_i}{\x}=g_i(\x)$ for all $\x\in L$ and all rows $i\in[m]$.
  To show that then $L$ must satisfy (\ref{eq:sol0}), take any two
  vectors $\x,\y\in L$ and assume that $D_i\x=D_i\y$.  This means that
  vectors $\x$ and $\y$ must coincide in all positions where the $i$th
  row of $A$ has stars. Since $g_i$ can only depend on these
  positions, this implies $g_i(\x)=g_i(\y)$, and hence,
  $\skl{\a_i}{\x}= \skl{\a_i}{\y}$.

  $(\Leftarrow)$ Assume that $L\subseteq\gf^n$ satisfies
  (\ref{eq:sol0}). We have to show that then there exists an operator
  $G=(g_1,\ldots,g_m)$ consistent with $A$ such that
  $\skl{\a_i}{\x}=g_i(\x)$ for all $\x\in L$ and $i\in[m]$; here, as
  before, $\a_i$ is the $i$th row of $A$ with all stars set to~$0$.
  The $i$th row of $A$ splits the set $L$ into two subsets
  \[
  L_i^{0}=\{\x\in L\colon\skl{\a_i}{\x}=0\}\ \ \mbox{and}\ \
  L_i^{1}=\{\x\in L\colon\skl{\a_i}{\x}=1\}\,.
  \]
  Condition (\ref{eq:sol0}) implies that $D_i\x\neq D_i\y$ for all
  $(\x,\y)\in L_i^0\times L_i^1$. That is, if $S_i$ is the set of
  star-positions in the $i$th row of $A$, then the projections
  $\proj{\x}{S_i}$ of vectors $\x$ in $L_i^0$ onto these positions
  must be different from all the projections $\proj{\y}{S_i}$ of
  vectors $\y$ in $L_i^1$.  Hence, we can find a boolean function
  $g_i:\gf^{S_i}\to\gf$ taking different values on these two sets of
  projections. This function will then satisfy
  $g_i(\x)=\skl{\a_i}{\x}$ for all $\x\in L$.
\end{proof}

A \emph{coset} of a set of vectors $L\subseteq\gf^n$ is a set
$\v+L=\{\v\+\x\colon \x\in L\}$ with $\v\in\gf^n$.  Since
$(\v+L)+(\v+L)=L+L$, Theorem~\ref{thm:cayley} implies:

\begin{cor}\label{cor:coset}
  Every coset of a solution for a $(0,1,\ast)$-matrix $A$ is also a
  solution for $A$.
\end{cor}

\begin{rem}\label{rem:cayley}\rm
  A Cayley graph over the Abelian group $(\gf^n,\+)$ generated by a
  set $K\subseteq \gf^n$ of vectors has all vectors in $\gf^n$ as
  vertices, and two vectors $\x$ and $\y$ are joined by an edge iff
  $\x\+\y\in K$.  Theorem~\ref{thm:cayley} shows that solutions for a
  $(0,1,\ast)$-matrix $A$ are precisely the independent sets in a
  Cayley graph generated by a special set~$K_A$.
\end{rem}

\begin{rem}\rm
  If $A$ is an $m$-by-$n$ $(0,1)$-matrix, that is, has no stars at
  all, then $K_A=\{\x\colon A\x\neq\nulis\}$. Hence, in this case, a
  set $L\subseteq\gf^n$ is a solution for $A$ iff there is a vector
  $\b\in\gf^m$ such that $A\x=\b$ for all $\x\in L$. That is, in this
  case, $\ker{A}=\{\x\colon A\x=\nulis\}$ is an optimal solution.
\end{rem}

\section{Structure of linear solutions}
\label{sec:linear}

By Theorem~\ref{thm:cayley}, a set of vectors $L\subseteq\gf^n$ is a
solution for an $m$-by-$n$ $(0,1,\ast)$-matrix $A$ if and only if
$(L+L)\cap K_A=\emptyset$, where $K_A\subseteq\gf^n$ is the set of
``forbidden'' vectors for $A$.  Thus, \emph{linear}
solutions are precisely vector subspaces of $\gf^n$ avoiding the set
$K_A$. Which subspaces these are? We will show (Theorem~\ref{thm:lin-str})
that these are precisely the subspaces lying entirely in the kernel of
some completion of~$A$.

Each vector subspace of $\gf^n$ is a kernel $\ker{H}=\{\x\colon
H\x=\nulis\}$ of some $(0,1)$-matrix $H$. Hence, linear solutions for
$A$ are given by matrices $H$ such that $H\x\neq \nulis$ for all
$\x\in K_A$; in this case we also say that the matrix $H$ \emph{separates}
$K_A$ from zero.  By the \emph{span-matrix} of a $(0,1)$-matrix $H$ we
will mean the matrix $\spnm{H}$ whose rows are all linear combinations
of the rows of~$H$.

\begin{lem}\label{lem:separ}
  Let $A$ be a $(0,1,\ast)$-matrix and $H$ be $(0,1)$-matrix. Then
$\ker{H}$ is a solution for $A$ iff $\spnm{H}$ contains a completion of~$A$.
\end{lem}

\begin{proof}
 To prove $(\Leftarrow)$, suppose that some completion $M$ of
  $A$ is a submatrix of $\spnm{H}$. Let $\x\in K_A$. By Lemma~\ref{lem:forb},
  we know
  that then $M\x\neq\nulis$, and hence, also $\spnm{H}\x\neq\nulis$.
  Since $H\x=\nulis$ would imply $\spnm{H}\x=\nulis$, we also have
  that $H\x\neq\nulis$.

  To prove $(\Rightarrow)$, suppose that $\ker{H}$ is a solution for $A$,
that is,  $H\x\neq \nulis$ for all $\x\in K_A$. Then,
  for every row $i\in[m]$ and every vector $\x\in\gf^n$, $H\x=\nulis$
  and $D_i\x=\nulis$ imply that $\skl{\a_i}{\x}=0$. This means that $\a_i$
must be a linear combination of rows of $H$ and $D_i$.
  Hence, for each $i$, the vector $\a_i$ must lie in the vector space
  spanned by the rows of $H$ and $D_i$, that is,
  $\a_i=\Alpha_i^{\top}H\+\Beta_i^{\top}D_i$ for some vectors
  $\Alpha_i$ and $\Beta_i$. In other words, the $i$th linear
  combination $\Alpha_i^{\top}H$ of the rows of $H$ is the $i$th row
  $\a_i\+\Beta_i^{\top}D_i$ of a particular completion $M$ of~$A$,
  implying that $M$ is a submatrix of~$\spnm{H}$, as desired.
\end{proof}

\begin{thm}\label{thm:lin-str}
  Let $A$ be a $(0,1,\ast)$-matrix. A linear subspace is a solution
  for $A$ if and only if it is contained in a kernel of some
  completion of $A$.
\end{thm}

\begin{proof}
  ($\Leftarrow$): If a linear subspace $L\subseteq\gf^n$ lies in a
  kernel of some completion of $A$ then $L\cap K_A=\emptyset$, by
 Lemma~\ref{lem:forb}. Since $L+L=L$, the set $L$ must be
 a solution for $A$, by Theorem~\ref{thm:cayley}.

  ($\Rightarrow$): Let $L\subseteq\gf^n$ be an arbitrary linear
  solution for $A$. Then $L+L= L$ and $L\cap K_A=\emptyset$.  Take a
  $(0,1)$-matrix $H$ with $L=\ker{H}$. Since $\ker{H}\cap
  K_A=\emptyset$, the matrix $H$ separates $K_A$ from
  zero. Lemma~\ref{lem:separ} implies that then $\spnm{H}$ must
  contain some completion $M$ of $A$.  But then
  $L=\ker{H}=\ker{\spnm{H}}\subseteq \ker{M}$, as claimed.
\end{proof}

\begin{cor}\label{cor:lin-mr}
  For any $(0,1,\ast)$-matrix $A$ we have that
  $\lin{A}=2^{n-\minrk{A}}$.
\end{cor}

\begin{proof}
  By Theorem~\ref{thm:lin-str}, $\lin{A}$ is the maximum of
  $|\ker{M}|=2^{n-\rk{M}}$ over all completions $M$ of $A$.  Since
  $\minrk{A}$ is the minimum of $\rk{M}$ over all completions $M$ of
  $A$, we are done.
\end{proof}

\begin{cor}[Alternative definition of min-rank]\label{cor:mr}
  For every $(0,1,\ast)$-matrix $A$ we have
  \[
  \minrk{A}=\min\{\rk{H}\colon \mbox{$H\x\neq \nulis$ for all
$\x\in K_A$}\}\,.
  \]
\end{cor}

\begin{proof}
  Let $R$ be the smallest possible rank of a $(0,1)$-matrix separating
  $K_A$ from zero.  To prove $\minrk{A}\geq R$, let $M$ be a
  completion of $A$ with $\rk{M}=\minrk{A}$. By
  Lemma~\ref{lem:forb}, the matrix $M$ separates $K_A$ form
  zero. Hence, $R\leq \rk{M}= \minrk{A}$.

  To prove $\minrk{A}\leq R$, let $H$ be a $(0,1)$-matrix such that
  $H$ separates $K_A$ form zero and $\rk{H}=R$.  By
  Lemma~\ref{lem:separ}, the matrix $\spnm{H}$ must contain a
  completion $M$ of $A$. Hence, $\minrk{A}\leq \rk{M}\leq
  \rk{\spnm{H}}=\rk{H}=R$.
\end{proof}

By Lemma~\ref{lem:forb}, the complement of $K_A$ is the union
of kernels $\ker{M}$ of
all completions $M$ of $A$. So, Theorems~\ref{thm:cayley} and \ref{thm:lin-str}
imply that a subset $L\subseteq\{0,1\}^n$ is:
\begin{itemize}
\item a solution for $A$ iff $L+L\subseteq
\bigcup\, \left\{\ker{M}\colon\mbox{$M$ is a completion of $A$}\right\}$;
\item a linear solution for $A$ iff $L\subseteq\ker{M}$ for some
completion $M$ of $A$.
\end{itemize}

\section{Structure of general solutions}
\label{sec:gen}

The following theorem says that non-linear solutions must be ``very
non-linear'': they cannot contain large linear subspaces.  Recall that
in Valiant's setting (cf. Lemma~\ref{lem:valiant}) we may assume that
each row of a $(0,1,\ast)$-matrix contains at most $s=n^{\delta}$
stars, where $\delta>0$ is an arbitrary small constant.
Define the co-distance of a vector space as the smallest weight 
of a non-zero vector in its orthogonal complement.

\begin{thm}\label{thm:struct2}
  Let $L\subseteq\gf^n$ be a solution for an $m$-by-$n$
  $(0,1,\ast)$-matrix $A$, and let $s$ be the maximum number of stars
  in a row of $A$.  If $L$ contains a subspace of co-distance at least
  $s+1$, then $L$ lies in a linear solution for~$A$.
\end{thm}

\begin{proof}
  Since $L$ is a solution for~$A$, $W$ is a linear solution for $A$ as
  well.  Hence, by Theorem~\ref{thm:lin-str}, $W$ is contained in a
  kernel of some completion $M$ of~$A$. Our goal is to show that then
  the entire solution $L$ must be contained in~$\ker{M}$.  To show
  this, we will use the following simple fact.

\begin{clm}\label{clm:rem1}
  Let $W\subseteq \gf^n$ be a linear subspace of co-distance at least $k+1$.
  Then, for every $k$-element subset $S\subseteq[n]$ and for every
  vector $\y\in \gf^n$, there is a vector $\x\in W$ such that
  $\x\neq\nulis$ and $\proj{\y}{S}=\proj{\x}{S}$.
\end{clm}
\begin{proof}[Proof of Claim]
The set of all projections of vectors in $W$ onto $S$ forms a linear 
subspace. If this subspace would be proper, then some non-zero vector, 
whose support lies in $S$, would belong to the orthogonal complement 
of $W$, a contradiction.
\end{proof}

Assume now that $L\not\subseteq \ker{M}$, and take a vector $\y\in
L\setminus \ker{M}$. Since $\y\not\in\ker{M}$, we have that
$\skl{\m_i}{\y}=1$ for at least one row $\m_i$ of $M$.  Let $S$ be the
set of star-positions in the $i$th row of $A$ (hence, $|S|\leq s$),
and let $\a_i$ be this row of $A$ with all stars set to~$0$. By
Claim~\ref{clm:rem1}, there must be a vector $\x\in W\subseteq
L\cap\ker{M}$ with $\proj{\y}{S}=\proj{\x}{S}$, that is,
$D_i(\x\+\y)=\nulis$.  But $\x\in\ker{M}$ implies that
$\skl{\m_i}{\x}=0$.  Hence, $\skl{\m_i}{\x\+\y}=
\skl{\m_i}{\x}\+\skl{\m_i}{\y}=\skl{\m_i}{\y}=1$. Since the vector
$\a_i$ can only differ from $\m_i$ in star-positions of the $i$th row
of $A$ and, due to $D_i(\x\+\y)=\nulis$, the vector $\x\+\y$ has no
$1$'s in these positions, we obtain that $\skl{\a_i}{\x\+\y}=1$.
Hence, the vector $\x\+\y$ belongs to $K_A$, a contradiction with
$\x,\y\in L$.

This completes the proof of Theorem~\ref{thm:struct2}.
\end{proof}

\section{Relation to codes}
\label{sec:codes}

Let $1\leq r< n$ be integers. A (binary) error-correcting code of
minimal distance $r+1$ is a set $C\subseteq\gf^n$ of vectors, any two
of which differ in at least $r+1$ coordinates. A code is \emph{linear}
if it forms a linear subspace over $GF_2$.  The question on how good
linear codes are, when compared to non-linear ones, is a classical
problem in Coding Theory. We now will show that this is just a special
case of a more general ``$\opt{A}$ versus $\lin{A}$'' problem for
$(0,1,\ast)$-matrices, and that Min-Rank Conjecture in this special
case holds true.

An $(n,r)$-\emph{code matrix}, or just an $r$-\emph{code matrix} if
the number $n$ of columns is not important, is a $(0,1,\ast)$-matrix
with $n$ columns and $m=(r+1){{n}\choose{r}}$ rows, each of which
consists of $n-r$ stars and at most one $0$.  The matrix is
constructed as follows.  For every $r$-element subset $S$ of
$[n]=\{1,\ldots,n\}$ include in $A$ a block of $r+1$ rows $\a$ with
$a_i=\ast$ for all $i\not\in S$, $a_i\in\{0,1\}$ for all $i\in S$, and
$|\{i\in S\colon a_i=0\}|\leq 1$. That is, each of these rows has
stars outside $S$ and has at most one $0$ within $S$. For $r=3$ and
$S=\{1,2,3\}$ such a block looks like
\[
\begin{pmatrix}
  1 & 1 & 1 & \ast & \cdots & \ast\\
  0 & 1 & 1 & \ast & \cdots & \ast\\
  1 & 0 & 1 & \ast & \cdots & \ast\\
  1 & 1 & 0 & \ast & \cdots & \ast
\end{pmatrix}\,.
\]
A Hamming ball around the all-$0$ vector $\nulis$ is defined by
\[
\Ball{r}=\{\x\in\gf^n\colon 0\leq|\x|\leq r\}\,,
\]
where $|\x|=x_1+\cdots+x_n$ is the number of $1$'s in $\x$.

\begin{obs}\label{obs:codes0}
  If $A$ is an $r$-code matrix, then
  $K_A=\Ball{r}\setminus\{\nulis\}$.
\end{obs}

\begin{proof}
  Observe that no vector $\x\in\gf^r$, $\x\neq\nulis$ can be
  orthogonal to all $r+1$ vectors $\vienas,
  \vienas\+\e_1,\ldots,\vienas\+\e_r$ in $\gf^r$ with at most one $0$.
  Indeed, if $\skl{\x}{\vienas}=0$ then $\skl{\x}{\vienas\+\e_i}=x_i$
  for all $i=1,\ldots,r$.  By this observation, a vector $\x$ belongs
  to $K_A$ iff there is an $r$-element set $S\subseteq[n]$ of
  positions such that $\proj{\x}{S}\neq\nulis$ and
  $\proj{\x}{\overline{S}}=\nulis$, that is, iff $\x\neq\nulis$ and
  $\x\in\Ball{r}$.
\end{proof}

\begin{obs}\label{obs:codes}
  If $A$ is an $(n,r)$-code matrix, then the solutions for $A$ are
  error-correcting codes of minimal distance $r+1$, and linear
  solutions for $A$ are linear codes.
\end{obs}

\begin{proof}
  We have $(L+L)\cap (\Ball{r}\setminus\{\nulis\})=\emptyset$ iff
  $|\x\+\y|\geq r+1$ for all $\x\neq\y\in L$, that is, iff every two
  vectors $\x\neq \y\in L$ differ in at least $r+1$ positions. Hence,
  every solution for an $r$-code matrix $A$ is a code of minimal
  distance at least $r+1$, and linear solutions are linear codes.
\end{proof}

\begin{lem}\label{lem:codes}
  For code matrices, the min-rank conjecture holds with a constant
  $\epsilon>0$.
\end{lem}

\begin{proof}
  Let $A$ be an $(n,r)$-code matrix; hence,
  $K_A=\Ball{r}\setminus\{\nulis\}$.  Set $t:=\lfloor(r-1)/2\rfloor$.
  Since $|\x\+\y|\leq 2t<r$ for all $\x,\y\in \Ball{t}$, the sum of
  any two vectors $\x\neq\y$ from $\Ball{t}$ lies in $K_A$, implying
  that $\Ball{t}$ is a clique in the Cayley graph generated by $K_A$.
  Since, by Remark~\ref{rem:cayley}, solutions for $A$ are independent
  sets in this graph, and since in any graph the number of its
  vertices divided by the clique number is an upper bound on the size
  of any independent set, we obtain:
  \begin{equation}\label{eq:hamming}
    \opt{A}\leq 2^n/|\Ball{t}|= 2^n\Big/\sum_{i=0}^t{{n}\choose{i}}\,,
  \end{equation}
  which is the well-known Hamming bound for codes.  On the other hand,
  Gilbert-Varshamov bound says that linear codes in $\gf^n$ of
  dimension $k$ and minimum distance $d$ exist, if
  \[
  \sum_{i=0}^{d-2}{{n-1}\choose{i}}<2^{n-k}\,.
  \]
  Hence,
  \begin{equation}\label{eq:GV}
    \lin{A}\geq 2^n\Big/\sum_{i=0}^{r}{{n}\choose{i}}\,.
  \end{equation}
  Together with (\ref{eq:hamming}), this implies that the inequality
  (\ref{eq:first}) holds with $\epsilon$ about $1/2$.
\end{proof}

The example of code matrices also shows that the gap between min-rank
and row/column min-rank may be at least logarithmic in~$n$.

\begin{lem}\label{lem:gap}
  If $A$ is an $(n,r)$-code matrix, then $\minrk{A}=\Omega(r\ln
  (n/r))$ but $\wminrkcol{A}\leq r+1$ and $\wminrkrow{A}\leq 2r$.
\end{lem}

\begin{proof}
  To prove $\minrk{A}=\Omega(r\ln (n/r))$, recall that
  $K_A=\Ball{r}\setminus\{\nulis\}$. Hence, Corollary~\ref{cor:mr}
  implies that $\minrk{A}$ is the smallest possible rank of a
  $(0,1)$-matrix $H$ such that
  $\ker{H}\cap\Ball{r}\subseteq\{\nulis\}$.  On the other hand, for
  any such matrix $H$, its kernel $L=\ker{H}$ is a (linear) code of
  minimal distance at least $r+1$ containing $|L|=2^{n-\rk{H}}$
  vectors. Since, by Hamming bound~(\ref{eq:hamming}), no code $L$ of
  distance at least $r+1$ can have more than $N=2^n/(n/r)^{O(r)}$
  vectors, we have that
  \[
  \rk{H}=n-\log_2|L|\geq n-\log_2N =\Omega(r \ln (n/r))\,.
  \]

  To prove that $\wminrkcol{A}\leq r+1$, suppose that $A$ contains
  some $m\times k$ submatrix $B$ of min-rank $k$. Since all $k$
  columns must be independent, at least one row $\b$ of $B$ must be
  $\ast$-free and contain an odd number $|\b|$ of $1$'s. But every row
  of $A$ (and hence, also $\b$) can contain at most one $0$, implying
  that $|\b|\geq k-1$. Together with $|\b|\leq r$, this implies that
  $k\leq r+1$.

To prove that $\wminrkrow{A}\leq 2r$, recall that each row of $A$
  consists of $n-r$ stars and at most one $0$; the remaining $r$ (or
  $r-1$) entries are $1$'s.  Suppose now that $A$ contains some set
  $X$ of $|X|=k+1$ independent rows. That is, no subset of
  these rows can be made linearly dependent by setting $\ast$'s to $0$
  or $1$. The rows in $X$ must be, in particular, \emph{pairwise}
  independent. This, in particular, means that the set $X$ can contain at
most one row without $0$-entries.
So, let $Y\subseteq X$ be a set of $|Y|=k$ rows
containing $0$-entries.
Take any two rows
$\x\neq \y\in Y$ with $x_i=0$ and $y_j=0$. Since $\x$ and $\y$ are
independent and have only $\ast$'s or $1$'s outside their $0$-entries,
we have that: $i\neq j$ and either $x_j=1$ or $y_i=1$.
This implies that
the total number of $1$'s in the rows of $Y$ must be at least the number
${{k}\choose{2}}$ of pairs of vectors in $Y$.
So, there must exist a row $x\in Y$ with
$|\x|\geq {{k}\choose{2}}/|Y|=(k-1)/2$. Together with  $|\x|\leq r-1$,
this implies that
  $k\leq 2r-1$, and thus, that $|X|=k+1\leq 2r$.
\end{proof}

\section{Conclusion and open problems}

In this paper we pose a conjecture about systems of semi-linear
equations and show its relation to proving super-linear lower bounds
for log-depth circuits. We then give a support for the conjecture by
proving that some its weaker versions are true.  We also show that
solutions are independent sets in particular Cayley graphs, thus
turning the conjecture in a more general (combinatorial)
setting. Using this, we prove several structural properties of sets of
solutions that might be useful when tackling the original conjecture.

We defined solutions for a given $m$-by-$n$ $(0,1,\ast)$-matrix $A$ as
sets $L\subseteq\gf^n$ of vectors $\x$ satisfying a system of
equations
\begin{equation}\label{eq:concl1}
  \skl{\a_i}{\x} = g_i (D_i\x)\qquad i=1,\ldots,m\,,
\end{equation}
where $\a_i$ is the $i$th row of $A$ with all stars replaced by $0$,
$g_i$ is an arbitrary boolean function, and $D_i$ is a diagonal
$n$-by-$n$ $(0,1)$-matrix corresponding to stars in the $i$th row of
$A$.  We have also shown (see Remark~\ref{rem:cayley}) that solutions
for $A$ are precisely the independent sets in a Cayley graph over the
Abelian group $(\gf^n,\+)$ generated by a special set of vectors
\begin{equation}\label{eq:concl2}
  K_A = \{ \x \colon \exists i\ D_i \x =\nulis \mbox{ and }
  \skl{\a_i}{\x} = 1\}\,.
\end{equation}
The following two questions about possible generalizations of the
min-rank conjecture naturally arise:
\begin{enumerate}
\item What if instead of diagonal matrices $D_i$ in (\ref{eq:concl1})
  we would allow other $(0,1)$-matrices?
\item What if instead of special generating sets $K_A$, defined by
  (\ref{eq:concl2}), we would allow other generating sets?
\end{enumerate}
The following two examples show that the min-rank conjecture cannot be
carried too far: its generalized versions are false.

\begin{ex}[Bad generating sets $K$]\label{ex:bad}\rm
  Let $G$ be a Cayley graph generated by the set $K\subseteq\gf^n$ of
  all vectors with more than $n-2\sqrt{n}$ ones.  If $L\subseteq\gf^n$
  consist of all vectors with at most $n/2-\sqrt{n}$ ones, then
  $(L+L)\cap K=\emptyset$, that is, $L$ is an independent set in $G$
  of size $|L|\geq 2^{n-O(\log n)}$.  But any \emph{linear}
  independent set $L'$ in $G$ is a vector space of dimension at most
  $n-2\sqrt{n}$. Hence, $|L'|\leq 2^{n-2\sqrt{n}}$, and the gap
  $|L|/|L'|$ can be as large as $2^{\Omega(\sqrt{n})}$.

  Note, however, that there is a big difference between the set $K$ we
  constructed and the sets $K_A$ arising form $(0,1,\ast)$-matrices
  $A$: generating sets $K_A$ must be almost ``closed downwards''.  In
  particular, if $\x\in K_A$ then \emph{all} nonzero vectors, obtained
  from $\x$ by flipping some even number of its $1$'s to $0$'s, must
  also belong to~$K_A$. Hence, this example does not refute the
  min-rank conjecture as such.
\end{ex}

\begin{ex}[Bad matrices $D_i$]\rm
  Let us now look what happens if we allow the matrices
  $D_1,\ldots,D_m$ in the definition of a system of semi-linear
  equations (\ref{eq:concl1}) be \emph{arbitrary} $n\times n$
  $(0,1)$-matrices.  A completion $M$ of $A$ can then be defined as a
  $(0,1)$-matrix with rows $\m_i=\a_i+\Alpha_i^{\top}D_i$.  Now define
  $\minrk{A|D_1, \ldots,D_r}$ as the minimal rank of such a completion
  of $A$.  Observe that this definition coincides with the ``old''
  min-rank, if we take the $D_i$'s to be the diagonal matrices
  corresponding the stars in the $i$th row of~$A$.

  However, Example~\ref{ex:bad} shows that the min-rank conjecture is
  false in this generalized setting.  To see why, we can define
  appropriate matrices $A,D_1, \ldots,D_m $ such that the
  corresponding set $K_A$ defined by (\ref{eq:concl2}) consists of
  vectors with more than $n - 2\sqrt{n}$ ones: for an arbitrary vector
  $\v$ with more than $n - 2\sqrt{n}$ ones just define $\a_i$ and
  $D_i$ such that the system $D_i \x = \nulis, \skl{\a_i}{\x}=1$ has
  $\v$ as its only solution.
\end{ex}

Except of the obvious open problem to prove or disprove the
linearization conjecture (Conjecture~\ref{conj:gen}) or the min-rank
conjecture (Conjecture~\ref{conj:goal}), there are several more
concrete problems.

We have shown (Lemma~\ref{lem:gap}) that the gap between min-rank and
row/column min-ranks may be as large as $\ln n$. It would be
interesting to find $(0,1,\ast)$-matrices $A$ with larger gap.
\begin{probl}\label{probl:gap}
  How large can the gap $\minrk{A}/\max\{\wminrkcol{A},
  \wminrkrow{A}\}$ be?
\end{probl}

The next question concerns the clique number $\cl{G_A}$ of (that is,
the largest number of vertices in) Cayley graphs $G_A$ generated by
the sets of the sets $K_A\subseteq\gf^n$ of the form
(\ref{eq:concl2}). By Remark~\ref{rem:cayley}, solutions for $A$ are
independent sets in this graph. Hence, $\opt{A}$ is just the
independence number $\alpha(G_A)$ of this graph. Since in any
$N$-vertex graph $G$ we have that $\cl{G}\cdot\alpha(G)\leq N$, this
yields $\opt{A}\leq 2^n/\cl{G_A}$. On the other hand, it is easy to
see that $\cl{G_A}\leq 2^{\rk{M}}$, where $M$ is a canonical
completion of $A$ obtained by setting all $\ast$'s to $0$: If
$C\subseteq \gf^n$ is a clique in $G_A$, then we must have $M\x\neq
M\y$ for all $\x\neq \y\in C$, because otherwise the vector $\x\+\y$
would not belong to~$K_A$.

\begin{probl}
  Give a lower bound on $\cl{G_A}$ in terms of min-rank $\minrk{A}$
  of $A$.
\end{probl}

Finally, it would be interesting to eliminate an annoying requirement
in Theorem~\ref{thm:weak-row} that the matrix $A$ must be
star-monotone.

\begin{probl}
  If $A$ is an $r$-by-$n$ $(0,1,\ast)$-matrix of min-rank $r$, is then
  $\opt{A}\leq 2^{n-r}$?
\end{probl}

\begin{table}[htbp]
  \caption{This table summarizes the concepts introduced in this paper.
    Here $A$ is a partially defined $m\times n$
    matrix with entries from $\{0,1,\ast\}$.}
  \begin{center}
    \begin{tabular}{p{0.3\textwidth}p{0.1\textwidth}p{0.5\textwidth}}
      \hline
      Concept &  Notation & Meaning\\
      \hline\hline
      & & \\
      Completion of $A$ & & A $(0,1)$-matrix obtained from $A$ by
      setting
      its $\ast$-entries to $0$ and~$1$.\\
      Canonical completion of $A$ & & All $\ast$-entries of $A$ set to $0$.\\
      Min-rank & $\minrk{A}$ & Minimal rank over $GF_2$ of a
      completion of $A$.
      \\
      Max-rank & $\maxrk{A}$
      &  Maximal rank over $GF_2$ of a completion of $A$.\\
      Operator $G$ consistent with $A$ & & The $i$th coordinate of
      $G:\{0,1\}^n\to\{0,1\}^m$ can only depend on variables
      corresponding to $\ast$-entries
      in the $i$th row of~$A$.\\
      Solution for $A$ & & A set $L\subseteq\{0,1\}^n$ of the form
      $L=\{\x\colon M\x=G(\x)\}$,
      where $M$ is a completion of $A$, and $G$ is an operator consistent with $A$.\\
      Linear solution for $A$ & & A solution for $A$ forming a linear
      subspace of
      $\{0,1\}^n$.\\
      &  $\opt{A}$ & Maximum size of a solution for $A$.\\
      &  $\lin{A}$ & Maximum size of a linear solution for $A$;
$\lin{A}=2^{n-\minrk{A}}$. \\
      Min-Rank Conjecture & & $\opt{A}\leq
      2^{n-\epsilon\cdot\minrk{A}}$ for a constant $\epsilon>0$.\\
      Independence of $(0,1,\ast)$-vectors & & Cannot be made linear
      dependent by setting
      $\ast$'s to constants.\\
      Row min-rank & $\wminrkrow{A}$ & Maximal number of independent rows.\\
      Column min-rank & $\wminrkcol{A}$ & Maximal number of independent columns.\\
      Incidence matrix of $\ast$'s & $D_i$ & Diagonal $(0,1)$-matrix
      with
      $D_i[j,j]=1$ iff $A[i,j]=\ast$.\\
      Set of forbidden vectors & $K_A$ & All vectors $\x\in\{0,1\}^n$
      such that $D_i\x=0$ and $\skl{\a_i}{\x}=1$, where $\a_i$ is the
      $i$th row of $A$ with all stars set to $0$. Main property: $L$ is
      a solution for $A$ iff
      $(L+L)\cap K_A=\emptyset$.\\
      \hline\hline
    \end{tabular}
  \end{center}
  \label{tab:1}
\end{table}

\end{document}